\documentclass[conference]{IEEEtran}

\usepackage{tikz}
\usepackage{subcaption}
\usepackage{amsmath,amssymb}
\usepackage{algorithm}
\usepackage{algpseudocode}
\usepackage{graphicx}
\usepackage{booktabs}
\usepackage{url}
\usepackage{xcolor}
\usepackage{hyperref}
\usepackage[table]{xcolor}
\definecolor{ourrow}{gray}{0.94}
\usepackage{amsthm}
\newtheorem{lemma}{Lemma}
\usepackage{booktabs}
\usepackage{textcomp}
\usepackage{mdframed}

\renewcommand{\footnoterule}{%
  \kern -3pt                        
  \hrule width 0.4\columnwidth      
  \kern 2.6pt                       
}
\newcommand{\blackcircledwhite}[1]{%
  \tikz[baseline=(char.base)]{
    \node[shape=circle, fill=black, text=white, inner sep=1pt] (char) {\bfseries #1};
  }%
}

\title{Trust-Aware Routing for Distributed Generative AI Inference at the Edge}

\author{
\IEEEauthorblockN{Chanh Nguyen, Erik Elmroth}
\IEEEauthorblockA{
Department of Computing Science, Ume{\aa} University, SE-90187, Sweden\\
Email: \{chanh, elmroth\}@cs.umu.se
}
}

\begin{document}
\maketitle

\thispagestyle{plain} 
\pagestyle{plain}

\begin{abstract}
Emerging deployments of Generative AI increasingly execute inference across decentralized and heterogeneous edge devices rather than on a single trusted server. In such environments, a single device failure or misbehavior can disrupt the entire inference process, making traditional best-effort peer-to-peer routing insufficient. Coordinating distributed generative inference therefore requires mechanisms that explicitly account for reliability, performance variability, and trust among participating peers.

In this paper, we present G-TRAC, a trust-aware coordination framework that integrates algorithmic path selection with system-level protocol design to ensure robust distributed inference.
First, we formulate the routing problem as a \textit{Risk-Bounded Shortest Path} computation and introduce a polynomial-time solution that combines trust-floor pruning with Dijkstra's search, achieving sub-millisecond median routing latency at practical edge scales, and remaining below 10 ms at larger scales.
Second, to operationally support the routing logic in dynamic environments, the framework employs a \textit{Hybrid Trust Architecture} that maintains global reputation state at stable anchors while disseminating lightweight updates to edge peers via background synchronization.

Experimental evaluation on a heterogeneous testbed of commodity devices demonstrates that G-TRAC significantly improves inference completion rates, effectively isolates unreliable peers, and sustains robust execution even under node failures and network partitions.

\end{abstract}


\section{Introduction}
\label{sec:intro}
\begin{figure*}[t]
\vspace{0.08in}
    \centering
    \begin{subfigure}[t]{0.56\textwidth}
        \centering
        \includegraphics[width=\linewidth]{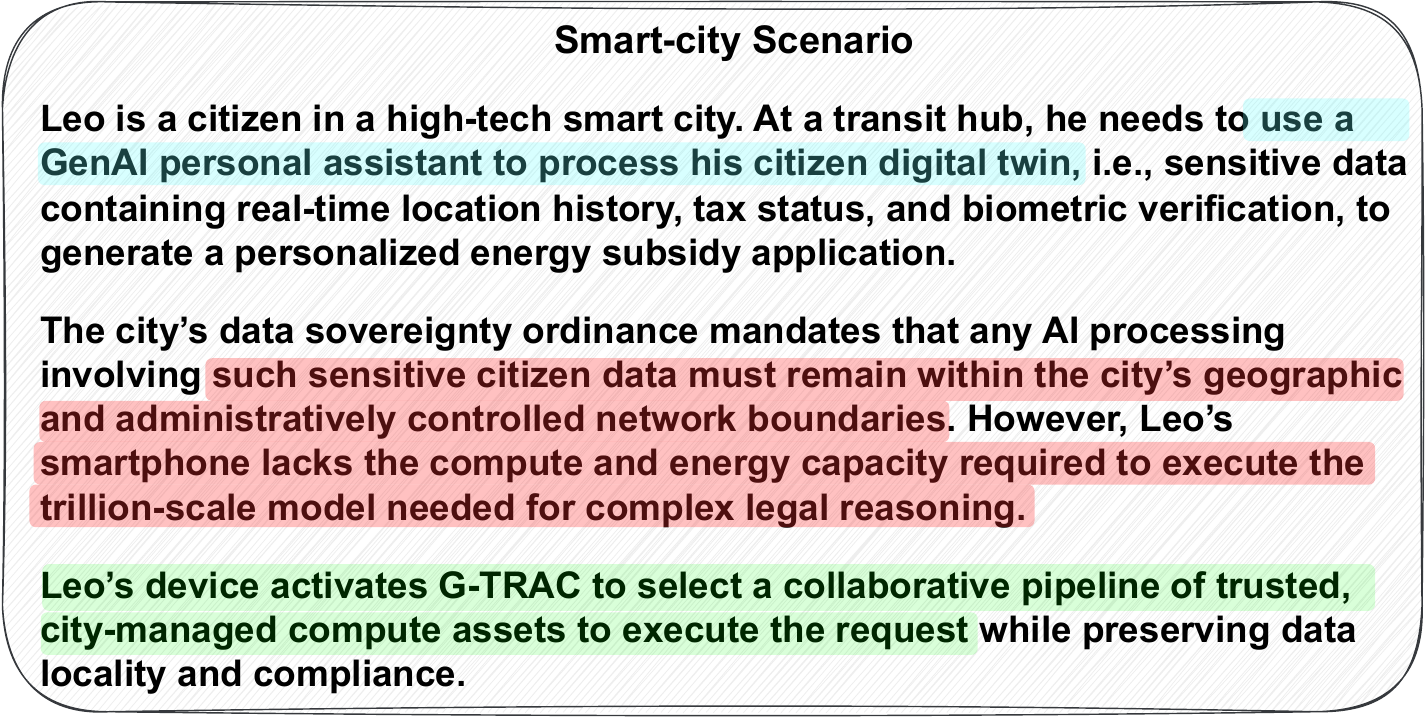}
        \caption{Motivating smart-city scenario with policy-constrained collaborative edge AI execution}

        \label{fig:subfig1}
    \end{subfigure}
    \hfill
    \begin{subfigure}[t]{0.35\textwidth}
        \centering
        \includegraphics[width=\linewidth]{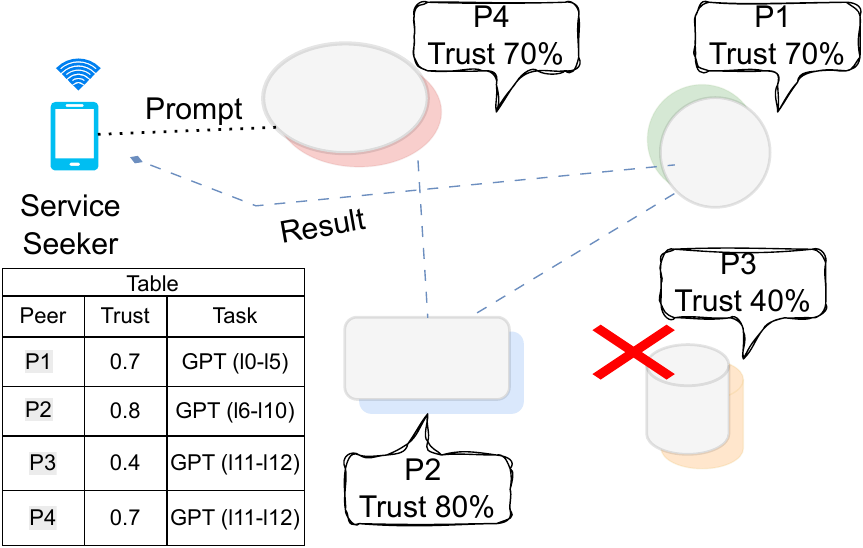}
        \caption{Trust-aware GenAI execution from the service seeker perspective using GTRAC}
        \label{fig:subfig2}
    \end{subfigure}
    \caption{Overview of a smart-city edge scenario and the role of G-TRAC in coordinating trust-aware, distributed GenAI execution.}

    \label{fig:scenario_tarp}
\end{figure*}

Generative artificial intelligence (GenAI) has recently emerged as a transformative class of models capable of producing human-like content (e.g., text, image, code, and video) from simple textual prompts \cite{cao2025survey, naveed2025comprehensive}. In contrast to traditional AI systems that primarily perform classification or prediction through a single forward pass \cite{lecun2015deep}, GenAI synthesizes new content through large-scale, autoregressive inference \cite{radford2018improving, vaswani2017attention}. 
These systems frequently utilize \textit{trillion-scale architectures} and \textit{Mixture-of-Experts} designs \cite{fedus2022switch, cai2025survey}, introducing significantly higher computational, memory, and latency demands due to the iterative and sequential nature of token generation.
To meet these demands, large-scale GenAI models are typically deployed on high-performance server clusters equipped with specialized AI accelerators and ultra-high-speed interconnects \cite{naveed2025comprehensive, aminabadi2022deepspeed, kwon2023efficient}.

Reliance on such centralized data centers exposes GenAI services to network disruptions, service outages, and increased \textit{privacy} and \textit{data-governance} complexity~\cite{hacker2023regulating}, particularly regarding sovereign data requirements, cross-jurisdictional handling~\cite{leboukh2023balancing, ye2024privacy}, and the risks inherent in large-scale data aggregation. These limitations have motivated increasing interest in \textit{edge-based GenAI deployment}, which promises lower network-induced latency, improved responsiveness, and greater resilience to disconnection. However, this promise is fundamentally constrained by the nature of edge environments themselves, i.e., edge devices are highly heterogeneous and severely resource-limited, making it infeasible for any single node (e.g., smartphones, IoT gateways, local edge servers) to host or execute large generative models in isolation.

To partially address this mismatch, Small Language Models (typically ranging from hundreds of millions to a few billion parameters)~\cite{wang2025comprehensive, magister2023teaching} have been developed for use on resource-constrained devices~\cite{qu2025mobile}. Nevertheless, a substantial performance gap remains: reducing model scale often degrades complex reasoning ability, factual accuracy~\cite{wei2022emergent}, and robustness against hallucinations~\cite{ji2023survey, kulkarni2025evaluating}. 

Although executing high-capacity models end-to-end remains infeasible on edge devices, recent advances in model partitioning and pipelined inference~\cite{ong2024efficient, ma2023poster, huang2019gpipe} enable cooperative execution of generative workloads across multiple resource-constrained nodes. In such settings, the central challenge extends beyond the technical partitioning of inference computation to ensuring robustness and reliability in the presence of stragglers and unreliable peers (e.g., nodes may become overloaded, disconnect without warning, or return incorrect or low-quality outputs). 
Existing peer-to-peer (P2P) frameworks~\cite{borzunov2023petals, gao2023gradientcoin, p2pllmedge} address these challenges only in a\textit{ best-effort} manner, i.e., decentralized inference execution without explicit latency or reliability guarantees, and lack a fundamental notion of computational trust that accounts for the behavioral history or observed reliability of autonomous edge participants.

We argue that, in decentralized edge environments, \textit{routing becomes a first-class control problem for generative inference workloads}, as execution must be coordinated across heterogeneous and unreliable peers rather than relegated to passive data transport~\cite{morabito2025smaller}. Unlike conventional networking, where packets follow stateless forwarding paths, GenAI inference routing must select execution paths under uncertainty, balancing latency, reliability, and dynamic availability. This requirement formulates the routing task as a Risk-Bounded Shortest Path (RBSP) problem, uniquely complicated by the iterative, token-based nature of the generation process.


To address these challenges, we present \textbf{G-TRAC} (\textbf{G}enerative \textbf{T}rust-Aware \textbf{R}outing and \textbf{A}daptive \textbf{C}haining), a coordination framework for distributed generative inference over decentralized edge networks. G-TRAC treats inference routing as a multi-objective control problem and dynamically selects execution paths across heterogeneous peers using lightweight trust and performance estimates. 

Our core contributions are as follows:

\begin{itemize}
    \item \textbf{Hybrid Trust Architecture:} 
     A coordination architecture that maintains globally consistent trust and reputation state at stable infrastructure anchors, while disseminating updates to resource-constrained edge participants via \textit{periodic background synchronization}, thereby decoupling control-plane latency from the inference critical path.

   \item \textbf{Efficient Risk-Bounded Routing:}
    A polynomial-time routing algorithm that solves the Risk-Bounded Shortest Path problem on a trusted subgraph. By combining trust-floor pruning with Dijkstra's algorithm, G-TRAC identifies optimal execution paths in sub-millisecond time, avoiding the combinatorial explosion typical of standard multi-constraint routing.

    \item \textbf{Physical Testbed Evaluation:}
    An end-to-end experimental evaluation on a heterogeneous edge testbed comprising commodity devices, demonstrating that G-TRAC improves inference completion rate, isolates underperforming or unreliable peers, and maintains robust GenAI service execution even under node failures and network partitions.
\end{itemize}

It is worth noting that our contributions focus on improving the reliability, robustness, and quality-of-experience of decentralized GenAI execution rather than on enhancing the semantic accuracy or reasoning capability of the underlying language models.
For clarity, Figure~\ref{fig:scenario_tarp} illustrates a representative edge scenario and the role of G-TRAC in coordinating trust-aware GenAI execution across heterogeneous peers.

\section{Related Work}

\begin{table*}[t]
\centering
\caption{Comparison of decentralized and cloud--edge GenAI inference systems}
\label{tab:related_work_comparison}
\renewcommand{\arraystretch}{1.3}
\begin{tabular}{|p{2.2cm}|p{3.8cm}|p{3.5cm}|p{3.0cm}|p{3.0cm}|}
\hline
\textbf{Work} &
\textbf{Primary Goal} &
\textbf{Routing / Coordination Logic} &
\textbf{Trust / Reliability Model} &
\textbf{Target Environment} \\
\hline
Petals (ACL’23)~\cite{borzunov2023petals} &
Democratize access to large models &
DHT-based discovery with latency-prioritized chain selection &
Best-effort (cooperative assumption)&
Consumer GPUs \\
\hline

Parallax~\cite{tong2025parallax} &
Maximize throughput under dynamic availability &
DAG-based scheduling with pipeline parallelism &
Best-effort (performance variability only) &
Decentralized GPU pools \\
\hline

DeServe~\cite{wu2025deserve} &
Affordable offline inference (batching) &
Cost-aware workload distribution &
Best-effort (cost-driven cooperation) &
Underutilized, preemptible GPUs in shared data centers
 \\
\hline

Helix (ASPLOS'25)~\cite{mei2025helix} &
Maximize throughput in heterogeneous clusters &
Max-flow / min-cost optimization  &
Managed infrastructure &
Heterogeneous GPU Clusters \\
\hline

HexGen (ICML'24)~\cite{hexgen} and HexGen-2 (ICLR'25)~\cite{jiang2025hexgen} &
Optimize disaggregated prefill/ decoding phases &
Graph partitioning \& Max-Flow scheduling &
Managed infrastructure &
Heterogeneous GPU Clusters \\
\hline

$\text{MoE}^2$ (MobiCom'25)~\cite{jin2025poster} &
Energy-accuracy expert optimization &
Expert ensemble selection&
Trust-agnostic optimization&
Edge Servers \\
\hline

Yang et. al. (TMC 2025)~\cite{yang2025quality} &
Optimize single-node expert selection &
Deep reinforcement learning &
Trust-agnostic optimization &
Mobile-to-Edge Offloading \\
\hline

Yu et. al.~\cite{yu2025efficient} &
Balance response time, cost, and quality &
NSGA-II Multi-objective genetic algorithm &
Trust-agnostic optimization &
Cloud--Edge Continuum \\
\hline

\rowcolor{ourrow}
\textbf{G-TRAC (Ours)} &
\textbf{Reliable, Risk-bounded} &
\textbf{Risk-bounded shortest path} &
\textbf{Explicit trust and risk modeling} &
\textbf{Heterogeneous IoT \& Mobile Edge} \\
\hline

\end{tabular}
\end{table*}

To enable large-model inference on hardware with limited memory, systems such as Petals~\cite{borzunov2023petals} adapt pipeline parallelism to public networks. These approaches overcome single-device resource constraints through techniques such as quantization and dynamic block loading to maximize peer participation. Recent advances such as Parallax~\cite{tong2025parallax} improve upon these baselines by introducing a two-phase scheduler that jointly optimizes model placement and runtime chain selection. Similarly, DeServe~\cite{wu2025deserve} leverages decentralized GPUs to reduce the cost of large-scale inference for offline workloads. However, these systems primarily optimize throughput and connectivity, rather than execution trust. They typically rely on greedy latency heuristics or global DAG planning that assumes cooperative peers. In adversarial or highly dynamic edge settings, where peers may be volatile or non-cooperative, such best-effort mechanisms leave the inference pipeline vulnerable to failures that are fatal for interactive workloads.

In contrast to volunteer networks, systems such as Helix~\cite{mei2025helix}, HexGen~\cite{hexgen}, and HexGen-2~\cite{jiang2025hexgen} focus on maximizing performance within controlled, heterogeneous clusters. HexGen enables asymmetric tensor and pipeline parallelism across diverse GPU generations, while HexGen-2 further optimizes execution through disaggregated prefill–decode scheduling. Helix formulates LLM serving as a max-flow optimization problem to maximize throughput. These systems implicitly assume a managed infrastructure with high node availability and trusted execution. Their reliance on heavyweight centralized orchestration (e.g., MILP solvers in Helix) makes them unsuitable for decentralized edge environments, where global state is difficult to maintain and peers operate autonomously.

Recent work has also explored QoS-aware routing of inference requests under quality, cost, or energy constraints. Efficient Routing~\cite{yu2025efficient} uses genetic algorithms (NSGA-II) to route requests to an optimal cloud or edge instance, while $\text{MoE}^2$~\cite{jin2025poster} employs a gating network to select suitable expert nodes under latency and energy budgets. Decentralized LLM Deployment~\cite{yang2025quality} reduces inference latency in mobile edge computing via parallel communication–computation protocols. While effective for selecting a single execution endpoint or optimizing local parallelism, these approaches do not address multi-hop routing risk. Generative inference involves a sequentially dependent execution chain, so selecting a \textit{good} node is insufficient if the path to or through it is untrusted.

G-TRAC differs fundamentally by treating generative inference routing as a risk-bounded shortest path problem, where routing decisions are constrained by empirically measured peer reliability. This enables robust completion of distributed inference pipelines even in decentralized, unreliable edge environments.

We summarize these state-of-the-art systems, along with our approach, in Table~\ref{tab:related_work_comparison}, highlighting differences in routing strategies, trust assumptions, and deployment environments.

\section{System Model and Problem Formulation}
\label{sec:system_model}
\subsection{Network and Entity Model}
We consider a decentralized edge computing environment represented by a directed overlay graph $\mathcal{G}=(\mathcal{V},\mathcal{E})$.
The node set $\mathcal{V}=\{A\}\cup\mathcal{S}\cup\mathcal{P}$ comprises three classes of entities:
\begin{itemize}
    \item \textbf{Anchor ($A$):} A stable, high-availability infrastructure node that maintains a global trust and reputation ledger. The Anchor serves solely as a control-plane coordinator and does not execute inference tasks or lie on the data path.

    \item \textbf{Compute Peers ($\mathcal{P}$):} A set of heterogeneous edge devices that contribute computational resources.
    Each peer $p_i \in \mathcal{P}$ is characterized by a dynamic trust score $r_i(t) \in [0,1]$ and an estimated execution latency $\hat{\ell}_i(t)$.
    Depending on the deployment model, a peer advertises either
    (i) a contiguous layer segment $[L^{\text{start}}_i, L^{\text{end}}_i]$ of a sharded model, or
    (ii) a specific pipeline stage $s_i$ in a functional inference pipeline.

    \item \textbf{Service Seekers ($\mathcal{S}$):} Resource-constrained edge devices (e.g., smartphones) that initiate GenAI inference requests but lack sufficient memory or compute capacity to execute the full workload locally.
\end{itemize}

The set of edges $\mathcal{E}$ represents feasible logical handovers between participating entities.
Specifically, a directed edge $(p_i, p_j) \in \mathcal{E}$ exists between two peers if and only if peer $p_i$ hosts model stage $k$ and peer $p_j$ hosts the subsequent stage $k+1$, enabling a valid sequential transfer.
Peers maintain liveness via periodic heartbeats to the Anchor, denoted by the indicator $a_p(t)\in\{0,1\}$.

\subsection{Workload Model: Generative Service Chains}
GenAI inference exhibits an inherent sequential dependency due to autoregressive decoding~\cite{vaswani2017attention}, where the generation of token $t$ strictly depends on token $t-1$.
In datacenter environments, systems mitigate this latency by exploiting intra-layer and pipeline parallelism across micro-batches~\cite{huang2019gpipe, jiang2024dynapipe}.
However, such dense parallelism is often infeasible in resource-constrained edge environments due to bandwidth and memory limitations~\cite{habibi2025edge}.
Consequently, edge-based inference naturally degrades into a distributed sequential process.

Accordingly, we model a GenAI inference request as a \textit{sequential service chain}
$\mathcal{T}=\langle \tau_1,\tau_2,\ldots,\tau_K\rangle$, where each stage $\tau_k$ represents a computation step with a strict execution dependency.
The output of $\tau_k$ serves as the input to $\tau_{k+1}$.
A request is executed along a physical path $\pi=\langle p^{(1)},\ldots,p^{(K)}\rangle$, where $p^{(k)} \in \mathcal{P}$ denotes the peer assigned to stage $k$.
At stage $k$, peer $p^{(k)}$ receives an intermediate state $x_{k-1}$ and computes $x_k = f_{p^{(k)}}(x_{k-1})$.

This abstraction unifies multiple GenAI deployment paradigms:
\begin{itemize}
    \item \textbf{Layer-Sharded Inference:} $x_k$ represents a hidden state (e.g., activation tensor) propagated between contiguous model layers hosted on different devices.
    
    \item \textbf{Agentic and Functional Pipelines:} $x_k$ represents a semantic artifact
    (e.g., a retrieval result, generated code fragment, or verification signal)
    passed between distinct models or tools.
\end{itemize}

The proposed routing and trust mechanisms operate solely on this sequential composition and are agnostic to the internal semantics of $x_k$. Without loss of generality, our subsequent analysis and evaluation focus on the layer-sharded inference model, as it presents the most stringent latency and bandwidth constraints.

\subsection{Risk and Reputation Model}
For each live peer $p$, the Anchor maintains the trust score $r_p(t)\in[0,1]$, updated based on the peer's historical performance. 
While peer failures may exhibit correlation due to shared infrastructure or load conditions, we assume that peer failures are conditionally independent given their trust scores to obtain a tractable baseline model.
Under this assumption, the reliability of a service chain $\pi$ is given by
\begin{equation}
\label{eq:rel_chain}
\mathrm{Rel}(\pi;t)=\prod_{p\in\pi} r_p(t)
\end{equation}

Accordingly, the risk of a service chain $\pi$ is defined as
\begin{equation}
\mathrm{Risk}(\pi;t)=1-\mathrm{Rel}(\pi;t)
\end{equation}

To guarantee service quality, any valid chain $\pi$ must satisfy $\mathrm{Risk}(\pi;t) \leq \epsilon$, where $\epsilon$ denotes the user-defined risk tolerance.

\subsection{Risk-Bounded Chain Selection Problem}
For each peer $p$, the Anchor observes the per-hop execution time $\ell^{\mathrm{obs}}_p(t)$ from completed inference requests, which captures the end-to-end delay incurred at peer $p$ (i.e., local computation, serialization, network transmission, and forwarding overhead).
The Anchor maintains a smoothed latency estimate using an Exponentially Weighted Moving Average (EWMA)~\cite{hunter1986exponentially}:

\begin{equation}
\hat{\ell}_p(t)
=
(1-\beta)\hat{\ell}_p(t-1)
+
\beta\,\ell^{\mathrm{obs}}_p(t)
\end{equation}
where $\beta\in(0,1)$ is the smoothing factor.

In volatile edge environments, minimizing raw execution time is insufficient if the selected path frequently fails. A fast peer with low reliability introduces significant tail latency due to timeouts and repair overhead. To capture this trade-off, we define the \textit{effective latency cost} $\mathcal{C}_p(t)$ as the expected time to complete a hop, explicitly accounting for failure probability:
\begin{equation}
\label{eq:effective_lat}
\mathcal{C}_p(t) = \hat{\ell}_p(t) + \big(1 - r_p(t)\big) \cdot  T_{\text{timeout}}
\end{equation}
where $T_{\text{timeout}}$ is a fixed system penalty reflecting failure detection and re-routing delay.
Accordingly, $\mathcal{C}_p(t)$ penalizes unreliable nodes, aligning the routing objective with the system's tail-latency goals.

Let $\Pi(t)$ denote the set of all live, contiguous service chains at time $t$. Given a user-defined risk tolerance $\epsilon\in(0,1)$, we formulate the \textit{Risk-Bounded Shortest Path} problem as finding the optimal chain $\pi^\star$ that minimizes the total effective latency while satisfying the global safety constraint:

\begin{equation}
\label{eq:rbsp_final}
\begin{aligned}
\pi^\star(t) = \arg\min_{\pi\in\Pi(t)}\quad & \sum_{p\in\pi} \mathcal{C}_p(t) \\
\text{s.t.}\quad & \prod_{p\in\pi} r_p(t) \ge 1 - \epsilon, \\
& a_p(t)=1 \quad \forall p\in\pi
\end{aligned}
\end{equation}

The optimization in Eq.~\eqref{eq:rbsp_final} combines an additive cost function with a multiplicative path constraint, a formulation corresponding to the classic \textit{Restricted Shortest Path} problem known to be NP-Hard~\cite{Garey_np}.
To address this computational intractability, Section~\ref{sec:gtrac_design} proposes a topological decomposition approach that enforces a strict local trust floor, reducing the problem to a polynomial-time shortest-path search on a trusted subgraph.

\section{G-TRAC Design and Algorithm}
\label{sec:gtrac_design}
In this section, we present G-TRAC, an efficient routing algorithm that solves the risk-bounded chain selection problem defined in Eq.~\eqref{eq:rbsp_final}. To operate under edge constraints and dynamic peer availability, G-TRAC employs a decentralized control loop that separates global state estimation at the Anchor from lightweight, seeker-side routing decisions.
Figure~\ref{fig:g-trac} illustrates the architecture and execution pipeline. The algorithm proceeds in three phases:
\blackcircledwhite{1} \textit{State Synchronization},
\blackcircledwhite{2} \textit{Constraint Pruning}, and
\blackcircledwhite{3} \textit{Feedback-Driven Execution}.

\begin{figure}[t]
\vspace{0.05in}
  \centering
  \includegraphics[width=\columnwidth]{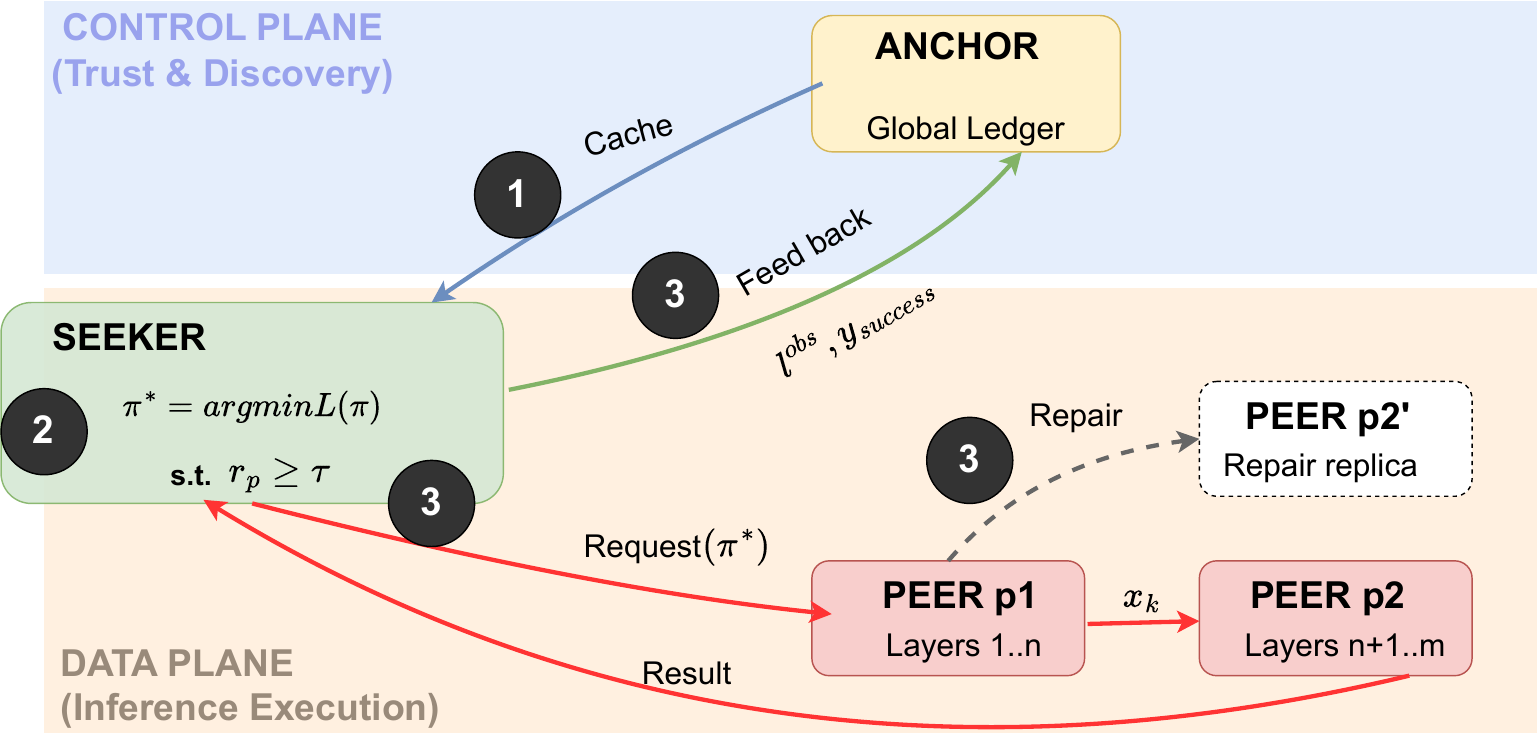}
  \caption{G-TRAC architecture and execution flow.}
  \label{fig:g-trac}
\end{figure}

\begin{algorithm}

\small
\caption{G-TRAC Seeker-Side Routing and Execution}
\label{alg:gtrac_opt}
\begin{algorithmic}[1]

\Require 
Token/layer sequence $\mathcal{T}$, trust threshold $\tau$, cached peer state $\tilde{\Sigma}$
\Ensure 
Execution outcome (\textsc{Success} or \textsc{Failure})

\Statex \textbf{//Trust-Based Pruning}
\State $V' \gets \{\, p \in \tilde{\Sigma} \mid a_p(t)=1 \land r_p(t) \ge \tau \,\}$

\Statex \textbf{//Latency-Optimal Path Search}
\State $G' \gets \textsc{BuildDAG}(V')$
\Comment{Construct DAG over contiguous layers}
\State $\pi^\star \gets \textsc{Dijkstra}(G', \text{weight}=\mathcal{C}_p)$
\Comment{Compute minimum-latency chain}
\If{$\pi^\star = \emptyset$}
    \State \Return \textsc{Abort}
    \Comment{No feasible contiguous chain}
\EndIf

\Statex \textbf{//Execution and Local Repair}
\State $\pi_{\textit{exec}} \gets \pi^\star$
\State $(\textit{res}, p_{\textit{fail}}) \gets \textsc{ChainExec}(\pi_{\textit{exec}})$

\If{$\textit{res} = \textsc{Failure}$ \textbf{and} \textsc{RepairEnabled}}
    \State $p_{\textit{new}} \gets 
    \arg\min\limits_{p \in V'} 
    \{ \hat{L}_p \mid p \neq p_{\textit{fail}} \land 
    \textsc{Layers}(p)=\textsc{Layers}(p_{\textit{fail}}) \}$
    
    \If{$p_{\textit{new}} \neq \emptyset$}
        \State $\pi_{\textit{exec}} \gets 
        \textsc{SwapNode}(\pi^\star, p_{\textit{fail}}, p_{\textit{new}})$
        \State $(\textit{res}, \_) \gets 
        \textsc{ChainExec}(\pi_{\textit{exec}})$
        \Comment{Retry with repaired chain}
    \EndIf
\EndIf
\Statex \textbf{//Feedback}
\State \textsc{UpdateTrust}(\textit{res}, $p_{\textit{fail}}$)
\State \Return \textit{res}

\end{algorithmic}
\end{algorithm}

\subsection{Phase 1: Asynchronous State Synchronization}
The Anchor acts as the control-plane authority, maintaining the global registry $\Sigma_t = \{ (p, \mathbf{c}_p, r_p, \hat{\ell}_p) \}_{p \in \mathcal{P}}$.
To avoid introducing control-plane latency into the inference critical path, Seekers do not query the Anchor synchronously per request.
Instead, they maintain a locally cached view $\tilde{\Sigma}_t \subseteq \Sigma_t$,
updated through periodic background synchronization.
As a result, routing decisions are computed locally without blocking on control-plane communication.

\subsection{Phase 2: Candidate Pruning via Trust Floors}

Directly enforcing the end-to-end constraint $\prod_{p\in\pi} r_p \ge 1-\epsilon$ typically requires evaluating complete paths, which is computationally expensive.
To make the search tractable, G-TRAC relaxes the global constraint into a local sufficiency condition, filtering the peer graph \textit{before} path selection.
We introduce a configurable \textit{trust floor} $\tau \in (0,1)$.

\begin{lemma}[Sufficiency of Local Trust]
\label{lem:trust}
For a service chain $\pi$ of length $K$, if every participating peer $p \in \pi$ satisfies $r_p(t) \ge \tau$, then the end-to-end failure risk is bounded by:
\begin{equation}
\mathrm{Risk}(\pi;t) \le 1 - \tau^K
\end{equation}
\end{lemma}
\emph{Proof.}
From Eq.~\eqref{eq:rel_chain}, $\mathrm{Rel}(\pi) = \prod_{p \in \pi} r_p \ge \prod_{p \in \pi} \tau
 = \tau^{|\pi|} = \tau^K$.
Thus, $\mathrm{Risk}(\pi) = 1 - \mathrm{Rel}(\pi) \le 1 - \tau^K$. \qed

The realized chain length $K$ is not known prior to routing because peers may host different numbers of model layers. However, $K$ is bounded by system design.

\textbf{Design Guarantee (Trust-floor configuration)}.
Let $K_{\max}$ denote the maximum feasible inference chain length, bounded by the model depth $L$ and the minimum number of layers per peer $l_{\min}$. If the trust floor is configured as
\[
\tau = (1 - \epsilon)^{1/K_{\max}}
\]
then any execution chain $\pi$ selected from the pruned graph satisfies $\prod_{p\in\pi} r_p \ge 1-\epsilon$
with respect to the current trust estimates. The proof is provided in Appendix~\ref{appendix:trust}.

This pruning step restricts the routing topology to the \textit{trusted subgraph} $\mathcal{G}'=(\mathcal{V}', \mathcal{E}')$, where $\mathcal{V}' = \{p \in \mathcal{V} \mid a_p(t)=1 \land r_p(t) \ge \tau\}$.
Consequently, the complex constrained optimization problem reduces to a standard shortest-path search over $\mathcal{G}'$.

\subsection{Phase 3: Feedback-Driven Execution}
Algorithm~\ref{alg:gtrac_opt} details the seeker-side procedure executed during this phase.

The Seeker first applies Phase~2 trust-floor pruning by excluding peers with zero liveness ($a_p = 0$) or insufficient trust ($r_p < \tau$) (line~1). 

From the remaining trusted peers, G-TRAC constructs a directed acyclic graph (DAG)
in which edges correspond to feasible handovers between contiguous model layers (line~2).
The Seeker then runs Dijkstra’s algorithm over the resulting subgraph (line~3) to identify the optimal execution chain $\pi^\star$ that minimizes the effective latency defined in Eq.~\eqref{eq:effective_lat}. 

During execution, G-TRAC implements a \textit{Bounded One-Shot Repair} policy to handle the stochastic nature of edge devices (line 9--15).
If execution fails at hop $p_k$, the Seeker does not discard the entire progress.
Instead, it queries the trusted set for a replacement peer $p'_k$ (where capabilities match and $r_{p'} \ge \tau$) and retries the failed step exactly once.

We explicitly bound this repair policy to a single attempt, as unbounded retries would obscure failure attribution and transform risk-bounded routing into probabilistic retry scheduling. This design choice aligns with established fault-tolerance practices in distributed systems (e.g., circuit breakers)~\cite{dean2013tail, beyer2016site},  where bounded retries are preferred to prevent cascading overload and to enable accurate failure detection.
Consequently, the one-shot mechanism does not aim to guarantee eventual success but serves as a \textit{bounded corrective action} that improves robustness against transient failures while preserving the semantics of risk-bounded routing and the integrity of trust learning.


Upon completion, the Seeker reports the execution trace to the Anchor for trust updates. G-TRAC applies targeted failure attribution: on success ($y=1$), all peers $p \in \pi$ receive a reward $\Delta r^{+}$, whereas on failure ($y=0$), only the peer responsible for the failed hop is penalized by $\Delta r^{-}$.




\subsection{Complexity Analysis}
Naively enumerating all feasible service chains induces a search space of size $\mathcal{O}(|\mathcal{P}|^{K})$, where $|\mathcal{P}|$ is the number of peers and $K$ is the chain length. This exponential growth makes exhaustive search intractable for large networks.
G-TRAC design reduces the selection process to polynomial time through two complementary
mechanisms:

\begin{enumerate}
    \item \textbf{Pre-search pruning}: Peers are filtered based on liveness $a_p(t)$
and the trust floor $\tau$. The filtering step runs in $\mathcal{O}(|\mathcal{P}|)$ time and reduces the effective graph size to $|V'| \ll |\mathcal{P}|$ by discarding unreliable
candidates prior to edge construction.

\item \textbf{Polynomial time selection}: The remaining candidates are organized into a DAG, and chain selection is formulated as a shortest-path problem solved using Dijkstra’s algorithm. The resulting complexity is bounded by $\mathcal{O}(|E'| + |V'|\log |V'|)$, where $|V'|$ and $|E'|$ denote the number of vertices and edges in the pruned graph.
\end{enumerate}

As a result, the complexity of G-TRAC depends on the size of the pruned network graph rather
than on the number of feasible service chains, with selection overhead remaining efficient even in dense networks.
The empirical validation of G-TRAC’s low-latency selection is presented in Section~\ref{subsecLoverhead}.


\section{Evaluation Methodology}
\label{sec:ex_setting}
\subsection{Experimental Testbed and Implementation}
We implemented a fully functional G-TRAC prototype using \texttt{PyTorch}\footnote{\url{https://pytorch.org/}}, and the \texttt{transformers}\footnote{\url{https://huggingface.co/docs/transformers/en/index}} library.
The testbed executes \textit{real-world distributed inference} where live tensors are serialized and transmitted between nodes over HTTP. 

\noindent \textbf{Hardware setup.}
To evaluate G-TRAC under realistic heterogeneity in compute capacity, we deploy a distributed testbed comprising a diverse set of edge nodes.
Table~\ref{tab:hardware_setup} summarizes the hardware characteristics of all nodes, including CPU type, core count, memory capacity, and the number of nodes of each type.
Android smartphones execute Linux user-space applications via Termux\footnote{\url{https://termux.dev/en/}}.

\begin{table}[t]
\vspace{0.05in}
\centering
\caption{Experimental testbed with heterogeneous nodes.}
\label{tab:hardware_setup}
\small
\setlength{\tabcolsep}{3pt}
\begin{tabular}{l l c c}
\toprule
\textbf{Type} & \textbf{Processor} & \textbf{Mem} & \textbf{\#} \\
\midrule
\(P\) (server) &
AMD Opteron 6272 &
47~GB & 4 \\

\(R\) (legacy server) &
Intel Xeon E5430 &
15~GB & 4 \\

\(D\) (desktop) &
Intel Core i5-13500 &
31~GB & 1 \\

\(L\) (laptop) &
Intel Core i7-1260P &
15~GB & 1 \\

\(S\) (phone) &
ARMv8 SoC (2×A76 + 6×A55) &
3.6~GB & 2 \\
\bottomrule
\end{tabular}
\end{table}






\noindent \textbf{Network overlay.}
The nodes are physically distributed across a campus environment and connect via a heterogeneous mix of wired Ethernet  and enterprise Wi-Fi (i.e., Eduroam).
Due to pervasive Network Address Translation (NAT) and inbound connectivity restrictions, direct peer-to-peer communication is not always feasible.
We therefore deploy an encrypted WireGuard-based mesh overlay using Tailscale\footnote{\url{https://tailscale.com/}}, which provides a flat IP addressing scheme and reliable bidirectional connectivity across all nodes.
The resulting overlay latency is included in all reported end-to-end and per-token latency measurements, reflecting realistic edge deployment conditions under constrained connectivity.


\noindent \textbf{GenAI workload.}
We use the pre-trained \textit{GPT-2 Large} model (774M parameters, 36 layers)~\cite{radford2019language} and partition it across worker nodes using pipeline parallelism.
Although larger transformer models exist, GPT-2 Large is sufficiently complex to be infeasible for stable monolithic execution on commodity edge devices (e.g., smartphones).

The workload consists of continuous autoregressive inference requests: a smartphone client submits a prompt that triggers token-by-token generation. Each generated token must sequentially pass through the distributed pipeline; consequently, delays or stragglers at any single peer directly impact the end-to-end generation latency.

\noindent \textbf{Peer heterogeneity and topology.}
We partition the model into contiguous shards\footnote{Detailed deployment plans and per-peer
configurations are provided at \url{https://github.com/anonymous-123qh/g-trac/tree/main/deployment}} (e.g., 3, 6, or 9 layers per shard) and deploy them across the physical machines described in Table~\ref{tab:hardware_setup}.

To emulate the scale of dense edge environments, we instantiate multiple virtual replicas per physical host.
These replicas are multiplexed over the same physical model shard to ensure realistic inference execution times, while their network behavior is software-defined to enforce distinct performance--reliability profiles:

\begin{itemize}
\item \textbf{Honey Pot (Risky--Fast):}
Peers offering ultra-low latency (added delay $\approx 1$\,ms) but exhibiting a high failure rate ($p_{\text{fail}} \in [0.20, 0.35]$), intentionally designed to expose the vulnerability of latency-greedy routing algorithms.

\item \textbf{Turtle (Safe--Slow):}
Peers providing near-perfect reliability ($p_{\text{fail}} \approx 0.1\%$) but at the cost of significantly higher network latency ($150\text{--}300$\,ms).

\item \textbf{Golden Peers (Guaranteed--Safe):}
Idealized peers with near-perfect reliability ($p_{\text{fail}} = 0$) and moderate latency ($20\text{--}40$\,ms).
\end{itemize}

In total, the testbed comprises 336 concurrent peers, forming a diverse routing search space that spans all stages of the inference pipeline.

\noindent \textbf{Trust and failure dynamics.}
To stress-test the convergence of the routing logic under dynamic conditions, we introduce controlled stochastic volatility using an additive asymmetric trust model.
Specifically, we emulate abstract peer failure behavior by allowing each peer $i$ to fail independently per request according to a Bernoulli random variable $X_i \sim \mathrm{Bernoulli}(p_{\text{fail}, i})$, where the failure probability $p_{\text{fail}, i}$ is determined by the peer's profile (e.g., high for Honey Pots, low for Turtles).
A failure stalls the request, preventing activation forwarding.

The complete set of system and workload parameters is summarized in Table~\ref{tab:sys-params}.

\begin{table}[t]
\vspace{0.05in}
\centering
\small
\caption{System parameters and configuration.}
\label{tab:sys-params}
\begin{tabular}{l|c|p{3.6cm}}
\toprule
\textbf{Param.} & \textbf{Value} & \textbf{Description} \\
\midrule
Model & GPT2-L & LLM 36 layers \\



$\tau$ & $0.96$ & Trust floor \\

$\beta$ & $0.30$ & Latency EWMA factor \\

$\ell_{init}$ & $250$ ms & Initial latency estimate \\

$\Delta r^{+}$ & $0.03$ & Trust reward \\

$\Delta r^{-}$ & $0.2$ & Trust penalty \\

$T_{hb}$ & $2$ s & Heartbeat interval \\

$T_{ttl}$ & $15$ s & Node timeout (liveness) \\

$T_{\text{timeout}}$ & $25$ s & Request timeout \\

$T_{gossip}$ & $2$ s & Registry sync period \\



\bottomrule
\end{tabular}
\end{table}

\subsection{Baselines}
\label{subsec:baselines}
We compare G-TRAC against four representative routing strategies that span the trade-off space between latency, reliability, and computational complexity:

\begin{itemize}
    
    \item  \textbf{Naive:} Enumerates feasible execution chains using Depth-First Search (DFS) before uniformly sampling a valid path. Unlike myopic random walks, which may terminate in local dead-ends, this approach guarantees the selection of a complete, end-to-end pipeline if one exists. In our practical implementation, we cap this search (e.g., at $1000$ chains) to bound the latency, though for scalability analysis, we evaluate the unbounded version to demonstrate its inherent complexity.

    \item \textbf{Shortest Path (SP):}
    Minimizes cumulative latency $\sum \hat{\ell}_p$ without enforcing trust constraints ( $\tau = 0$).
    SP represents traditional performance-centric routing protocols that optimize for speed but do not account for node reliability or failure risks.

    \item \textbf{Max-Reliability (MR):}
    Maximizes end-to-end reliability ($\prod r_p$) while disregarding latency cost.
    MR represents a conservative, risk-averse strategy that prioritizes fault tolerance over execution speed.

    \item \textbf{Lagrangian Relaxation (LARAC):}
    A standard heuristic for the Constrained Shortest Path problem~\cite{juttner2001lagrange} that iteratively optimizes a composite objective $L(\pi) + \lambda \cdot R(\pi)$.
    LARAC provides a strong theoretical baseline that attempts to balance latency and reliability, though its iterative nature incurs higher computational complexity than single-pass algorithms.
\end{itemize}

\subsection{Key Metrics}
We evaluate the comparative performance of G-TRAC and the baseline strategies using the following metrics:

\begin{itemize}
   \item \textbf{Service Success Rate (SSR):}
    The fraction of inference \emph{requests} that successfully complete and generate the full target token sequence without an unrecoverable failure:
    $
    \text{SSR} = \frac{N_{\text{success}}}{N_{\text{total}}}.
    $
    SSR captures the effective reliability of the routing policy.

   \item \textbf{Per-Token Latency:} The distribution of per-token end-to-end latency measured over tokens from successful requests, capturing typical token emission dynamics and variability during generation.

    \item \textbf{Selection Overhead:} The client-side wall-clock time required to execute the pruning and path selection logic.
\end{itemize}

\section{Result and Discussion}
\label{sec:result}

\subsection{Service Success Rate Analysis}
We evaluate the robustness of each routing strategy using SSR metric under prompt generation tasks with varying target lengths ($L_\text{tok} \in \{10, 20, 50\}$ tokens). For each configuration (algorithm, $L_\text{tok}$), we perform 100 independent prompt-generation requests. To ensure fair comparison across routing strategies, we reset peer trust states between algorithms.
 Figure~\ref{fig:ssr_tokens} reports the resulting SSR with error bars denote 95\% Wilson confidence intervals (CI)~\cite{wilson1927probable}.

\begin{figure}[t]
  \centering
  \includegraphics[width=0.8\columnwidth]{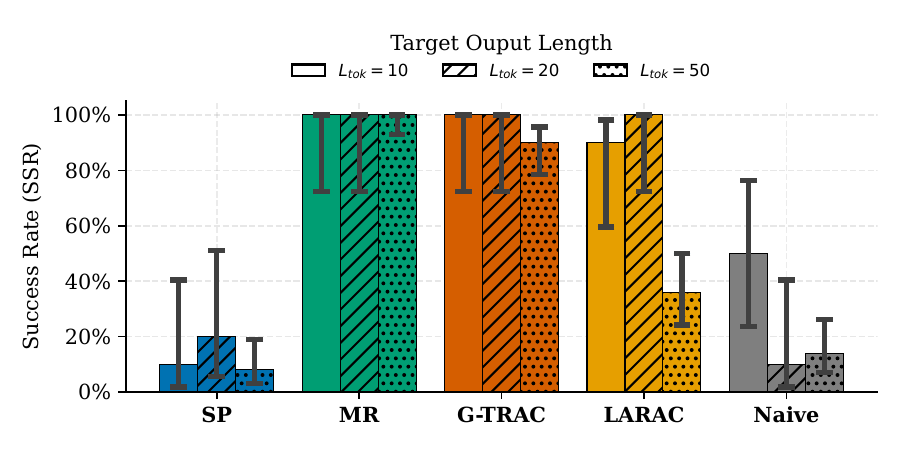}
  \caption{Service success rate  under different generation lengths.}
  \label{fig:ssr_tokens}
\end{figure}



We observe that both SP and Naive perform poorly. SP consistently fails to achieve success rates above 20\%, even for short sequences. This consistent failure is characteristic of the \textit{honey-pot} effect: by aggressively minimizing reported latency, SP preferentially selects fast but unreliable peers, leading to frequent chain breakages and a substantial degradation in SSR. Similarly, while the Naive approach achieves moderate success ($\approx 50\%$) for short contexts ($L_{\text{tok}}=10$), it exhibits high variance and fails to scale. For longer sequences, Naive's performance collapses to $< 15\%$, confirming that simple greedy heuristics lack the foresight required for long-horizon distributed inference.

In contrast, the MR baseline consistently achieves a 100\% SSR across all sequence lengths, validating that prioritizing peers solely on reliability maximizes service completion. G-TRAC demonstrates comparable robustness, maintaining perfect success for shorter sequences and degrading only marginally at $L_{\text{tok}}=50$. The overlap in the Wilson CI for G-TRAC and MR suggests that the two methods offer statistically comparable reliability. These observations highlight the effectiveness of the joint optimization of trust and latency, enabling G-TRAC to systematically avoid the risky honey pot peers that cause SP to fail.

Finally, LARAC exhibits significant sensitivity to task length. While it achieves $100\%$ SSR at $L_{\text{tok}}=20$, its performance significantly degrades to below $40\%$ at $L_{\text{tok}}=50$. The non-overlapping CIs confirm this degradation is statistically significant, suggesting that although LARAC is effective for short sequences, its rigid constraint enforcement struggles with longer, resource-intensive request chains.

\subsection{Per-Token Latency Analysis}
\begin{figure}[t]
  \centering
  \includegraphics[width=0.8\columnwidth]{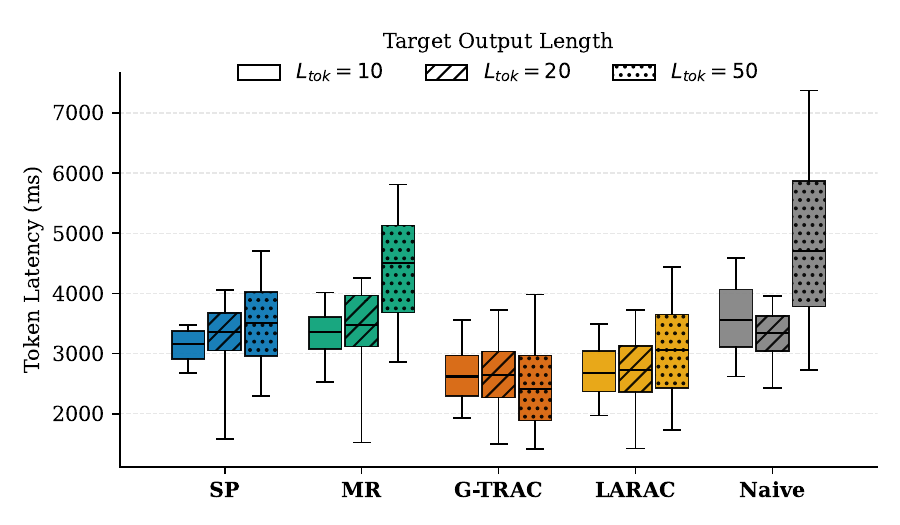}
 \caption{Per-token end-to-end latency distribution under different generation lengths. Black diamonds denote the P99 latency.}
  \label{fig:latency_tokens}
\end{figure}

Figure~\ref{fig:latency_tokens} shows the distribution of per-token end-to-end latency for successful inference requests. Across all routing strategies, per-token latency increases as the target generation length grows from $L_\text{tok}=10$ to $50$, reflecting inherent GPT-2 scaling effects, where longer contexts increase self-attention overhead due to larger KV caches and more expensive matrix multiplications.

Despite explicitly optimizing for minimum path latency, the SP baseline exhibits unexpectedly high per-token latency in completed runs (average latency of 3.1\,s with P99 of 3.6\,s at $L_\text{tok}=10$, increasing to 3.5\,s average and 5.1\,s P99 at $L_\text{tok}=50$). This counterintuitive behavior is explained by \emph{survivorship bias}, i.e., although SP aggressively selects low-latency paths, these paths are often unstable and fail before completion. Consequently, the successful SP executions, especially for longer sequences, are dominated by runs that traverse slower but more stable nodes, inflating the observed per-token latency.

The Naive baseline performs the worst overall, exhibiting the highest per-token latency across all generation lengths (e.g., an average latency of 3.6\,s with P99 of 4.8\,s at $L_\text{tok}=10$, increasing to 4.8\,s average latency and P99 up to 8.0\,s at $L_\text{tok}=50$). By selecting peers uniformly at random, Naive fails to avoid stragglers and high-latency nodes, resulting in consistently poor token-level performance. 

MR achieves perfect SSR, as discussed earlier, but at the cost of significantly increased per-token latency (from an average of 3.3\,s with P99 of 4.1\,s at $L_\text{tok}=10$ to 4.4\,s average and P99 of 6.0\,s at $L_\text{tok}=50$). This confirms that prioritizing highly reliable peers effectively reduces chain failures, albeit by routing through slower nodes that degrade latency performance.

Finally, both G-TRAC and LARAC demonstrate consistently lower per-token latency across all generation lengths. G-TRAC achieves an average per-token latency of 2.6\,s with P99 of 3.8\,s at $L_\text{tok}=10$, remaining stable at 2.4\,s average with P99 of 4.6\,s at $L_\text{tok}=50$. Similarly, LARAC maintains low latency, increasing from an average of 2.7\,s with P99 of 3.8\,s at $L_\text{tok}=10$ to 3.0\,s average and P99 of 5.0\,s at $L_\text{tok}=50$. This advantage stems from their constraint-aware optimization strategies, i.e., both approaches actively prune high-latency stragglers from the candidate pool, preventing slow nodes from dominating the execution chain. 

\begin{figure}[t]
  \centering
  \includegraphics[width=0.6\columnwidth]{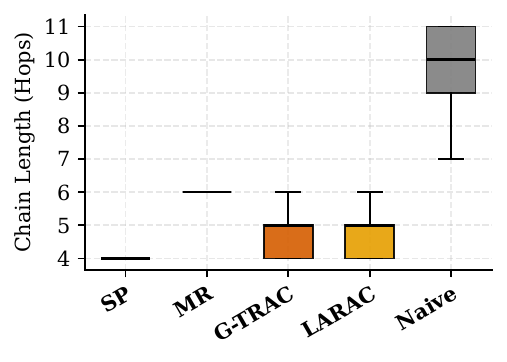}
  \caption{Distribution of inference chain length (hop count).}
  \label{fig:hop_count}
\end{figure}

\subsection{Adaptive Topology: Chain Length Distribution}

Figure~\ref{fig:hop_count} shows the distribution of inference chain lengths across routing strategies. SP exhibits zero variance, consistently selecting a 4-hop chain for all requests. With a model depth of 36 layers and a maximum per-peer capacity of approximately 9 layers, SP minimizes hop count by maximizing shard size per node. While topologically optimal, this rigid strategy concentrates computation on a small number of peers, contributing to the elevated per-token latency observed in Figure~\ref{fig:latency_tokens}.

In contrast, G-TRAC and LARAC center around a median of 4 hops but exhibit a tail extending to 5–6 hops. Both approaches adaptively select minimal-hop paths when sufficiently trusted peers are available, while dynamically extending the chain to avoid unstable nodes when necessary. Consequently, they trade a modest increase in hop count for improved reliability and lower realized latency.

By comparison, MR favors longer chains (6 hops), consistent with its conservative, reliability-first design. Naive produces highly variable and often excessively long chains, reflecting the absence of informed routing heuristics.

\subsection{The Peer Selection Landscape}
\begin{figure}[t]
  \centering
  \includegraphics[width=0.7\columnwidth]{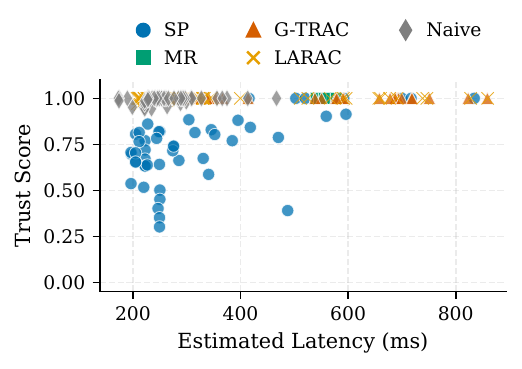}
\caption{Distribution of selected peers by Trust and Latency ($L_{tok}=50$).}
  
  \label{fig:peer_selection_landscape}
\end{figure}
To explain the divergent behaviors of the evaluated strategies, we visualize the \emph{peer selection landscape}, i.e., the trust--latency combinations of the peers chosen by each algorithm. Figure~\ref{fig:peer_selection_landscape} plots  the distribution of selected peers for 50-token generation tasks, with estimated latency on the $x$-axis and trust score on the $y$-axis.

SP concentrates on the low-latency region, including many low-trust peers. Since SP optimizes only for latency, it is frequently attracted to unreliable peers, which is consistent with its reduced success performance.

MR selects predominantly high-trust peers across a broader latency range. By prioritizing trust above all else, MR often avoids medium-latency peers when they incur even a small trust penalty, explaining its strong reliability but higher end-to-end latency.

G-TRAC and LARAC more consistently target the intermediate region that balances trust and latency. In particular, G-TRAC favors peers that retain high trust while avoiding MR's latency overhead. LARAC exhibits a more dispersed pattern, reflecting its iterative Lagrangian updates when enforcing the delay constraint.

\subsection{Routing Decision Overhead at Scale}
\label{subsecLoverhead}

\begin{figure}[t]
  \centering
  \includegraphics[width=0.7\columnwidth]{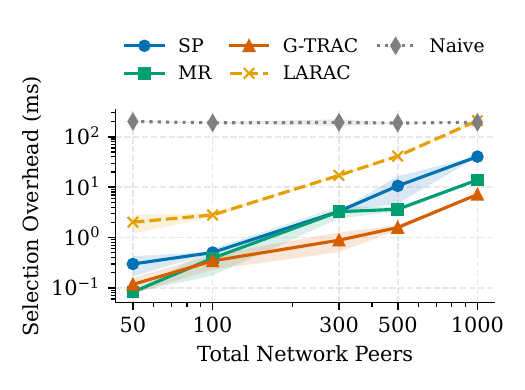}
  \caption{Algorithm scalability: Decision time vs Network size.}
  \label{fig:overhead}
\end{figure}

We evaluate the computational overhead of G-TRAC’s routing decision logic and compare it against representative baselines. To isolate control-plane costs, we decouple path selection from inference execution and implement exact versions of all routing algorithms. Experiments are conducted on a smartphone S, with network size varying from $N \in \{50,\dots,1000\}$ and results averaged over 100 independent trials.

Figure~\ref{fig:overhead} shows the selection overhead trends on a logarithmic scale for all approaches. G-TRAC consistently maintains negligible latency ($< 10$\,ms), even at $N=1000$ nodes. Notably, G-TRAC achieves lower selection overhead than both SP and MR. This is because while SP and MR operate over the full network graph $\mathcal{G}$, G-TRAC first constructs a pruned subgraph $\mathcal{G}' \subset \mathcal{G}$ by excluding nodes whose trust falls below the admissible threshold, substantially reduces the effective search space ($|V'| \ll |V|$, $|E'| \ll |E|$), allowing Dijkstra’s algorithm to converge faster.

In contrast, LARAC exhibits poor scaling (approximately $210$\,ms at $N=1000$) driven by iterative relaxation, while Naive becomes infeasible in dense networks ($>2$\,s timeout) due to DFS combinatorial explosion.

Overall, the results empirically validate that G-TRAC introduces minimal control-plane overhead while improving both reliability and scalability, enabling fully decentralized GenAI routing on resource-constrained edge devices.


\section{Conclusion}
In this paper, we presented G-TRAC, a coordination framework that elevates distributed genAI inference routing from a best-effort transport problem to a risk-bounded control problem. By integrating a hybrid trust architecture with a pruned shortest-path algorithm, G-TRAC decouples global state estimation from the critical routing path, enabling resource-constrained devices to make sub-millisecond routing decisions without centralized bottlenecks.

Empirical evaluation on a heterogeneous edge testbed demonstrates that G-TRAC consistently operates in a favorable trade-off region between trust and latency. In particular, it mitigates the honey-pot effect that degrades latency-driven baselines, achieving a near-perfect success rate for standard edge interaction workloads while preserving low end-to-end latency. 

Future work will extend G-TRAC to dynamic mixture-of-agents workflows, in which functional sub-tasks (e.g., retrieval, reasoning, and verification) are orchestrated across a mesh of untrusted, specialized peers, and will investigate the establishment of a principled root of trust to safely and efficiently harness the largely untapped computational capacity of the network edge.

\appendix
\subsection{Proof of the Design Guarantee}
\label{appendix:trust}

\begin{proof}
From Lemma~\ref{lem:trust}, any service chain $\pi$ of length $K$ composed of peers
satisfying $r_p(t) \ge \tau$ has reliability
$
\prod_{p \in \pi} r_p(t) \ge \tau^K 
$.

Since $0 < \tau \le 1$, the function $\tau^x$ is monotonically decreasing in $x$.
By system design, the realized chain length satisfies $K \le K_{\max}$, where
\[
K_{\max} = \left\lceil \frac{L}{l_{\min}} \right\rceil 
\]
Therefore,
\[
\prod_{p \in \pi} r_p(t) \ge \tau^{K_{\max}} 
\]
Substituting $\tau = (1-\epsilon)^{1/K_{\max}}$ yields
\[
\prod_{p \in \pi} r_p(t) \ge (1-\epsilon)^{K_{\max}/K_{\max}} = 1-\epsilon 
\]
Hence, any execution chain selected from the pruned graph satisfies the global
reliability constraint with respect to the current trust estimates.
\end{proof}
\subsection{Distributed GenAI Feasibility at the Edge}

\begin{figure}[t]
\vspace{0.05in}
    \centering
    \subfloat[Per token latency]{%
        \includegraphics[width=0.23\textwidth]{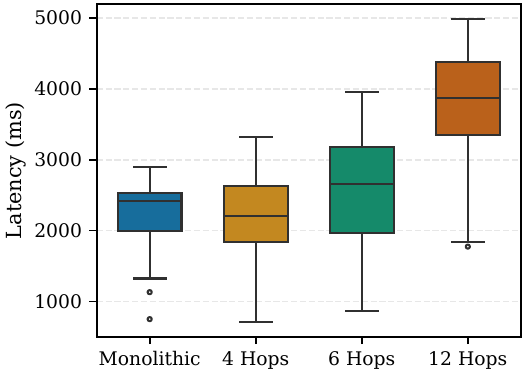}
        \label{fig:inference}
    }\hfill
    \subfloat[Per token CPU time]{%
        \includegraphics[width=0.23\textwidth]{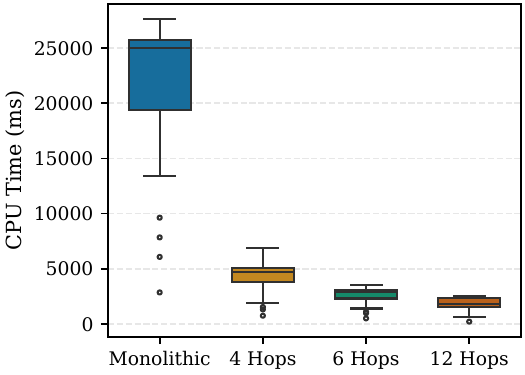}
        \label{fig:cputime}
    }
   \caption{Per-token time analysis: (a) end-to-end inference latency per token and (b) CPU execution time per token.}

    \label{fig:time}
\end{figure}

\begin{figure}[t]
    \centering
    \subfloat[CPU utilization]{%
        \includegraphics[width=0.23\textwidth]{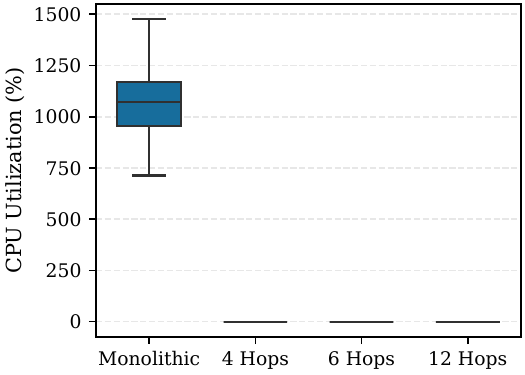}
        \label{fig:cpu}
    }\hfill
    \subfloat[Memory footprint]{%
        \includegraphics[width=0.23\textwidth]{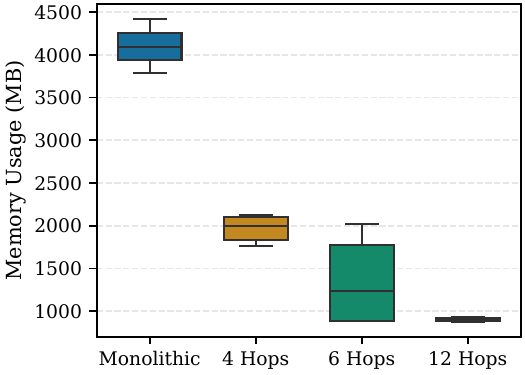}
        \label{fig:mem}
    }
    \caption{Resource consumption for monolithic vs. distributed configurations.}
    \label{fig:resource}
\end{figure}

We empirically characterize the feasibility and overheads of distributed GenAI model inference on heterogeneous edge devices.

To this end, we compare monolithic inference of GPT-2 Large, executed entirely on a single server-class node (\(P\)), against distributed configurations in which the model is partitioned into shards of 3, 6, and 9 layers. These partition sizes correspond to inference chains of 12, 6, and 4 hops, respectively.

Because some newer nodes in the testbed may offer stronger single-node compute performance than \(P\), this comparison is intended to characterize feasibility and deployment overheads in a heterogeneous setting rather than to establish an optimal single-node performance baseline.

Figure~\ref{fig:inference} presents the inference time per token for monolithic and distributed execution across varying hop lengths.
Under the observed network conditions, the latency overhead is highly dependent on chain length. For shorter chains (4 hops), distributed latency (approximately 2.2\,s) remains comparable to the monolithic baseline (approximately 2.3\,s). However, as the chain length increases to 12 hops, cumulative overheads from serialization and network transmission result in a latency increase of approximately $\mathbf{1.7\times}$ (approximately 3.8\,s). 

We further examine per-token CPU time to characterize compute demand. As presented in Figure~\ref{fig:cputime}, monolithic GPT-2 Large inference requires approximately 22\,s of CPU time per generated token on a single device, severely limiting sustained inference and co-location with other workloads. In contrast, distributed execution reduces the per-device CPU time to 1.8 - 4.4\,s per token, depending on shard size. This massive reduction in per-device compute demand enables inference to operate within realistic edge compute budgets, allowing devices to participate in inference without stalling other local background processes.

Finally, we analyze the resource consumption for monolithic and distributed configurations. Figure~\ref{fig:cpu} illustrates the CPU utilization. Monolithic inference nearly saturates available compute resources, averaging 953\% utilization (i.e., spanning multiple cores) with peaks exceeding 1400\%. In contrast, the distributed configuration maintains a minimal CPU footprint on individual peers, with average utilization remaining well below 10\% across all shard sizes, reflecting the sparse, intermittent nature of processing in the distributed chain.

Figure~\ref{fig:mem} depicts the peak Resident Set Size (RSS). Monolithic execution consistently demands substantial memory, with a mean footprint of 3.88\,GB and P95 peaks reaching 4.38\,GB. This renders stable deployment impractical on standard edge hardware (e.g., devices with 4GB unified memory limits). Distributed execution substantially lowers this barrier: peers hosting 9 layers require approximately 1.85\,GB, while 3-layer peers consume approximately 900\,MB. This reduction enables deployment on resource-constrained nodes and directly mitigates the memory capacity footprint challenges identified for AI agent inference~\cite{zhao2026heterogeneous}.

In summary, distributed GenAI makes large-model inference feasible in a heterogeneous edge environment by alleviating per-node memory and compute bottlenecks, at the cost of a bounded increase in end-to-end latency that grows with hop count.

\bibliographystyle{IEEEtran}
\bibliography{acmart}

\end{document}